\documentclass[a4paper,11pt]{article}
\usepackage[left=3.17cm,right=3.17cm,top=2.54cm,
headheight=0.5cm,headsep=0.54cm,bottom=2.54cm,footskip=0.79cm
]{geometry}

\usepackage[all]{xy}
\usepackage{amsmath,amsfonts,amssymb,amsthm,epsfig,amscd,comment,latexsym,psfrag}
\usepackage{nicefrac,xspace,tikz}

\usetikzlibrary{arrows}
\usetikzlibrary{decorations.markings}
\def\Z{\mathbb Z}

\def\C{\mathbb C}

\def\1{{\bf{1}}}

\usepackage[pdfstartview=FitH]{hyperref}

\def\footnoterule{\kern 1mm \hrule width 7cm \kern 2.2mm}%

 \newtheorem{thm}{Theorem}[section]
 \newtheorem{prp}[thm]{Proposition}
 \newtheorem{lem}[thm]{Lemma}
 \newtheorem{cor}[thm]{Corollary}

\newcommand{\bea}{\begin{eqnarray}}
\newcommand{\eea}{\end{eqnarray}}
\newcommand{\be}{\begin{equation}}
\newcommand{\ee}{\end{equation}}

\begin{document}

\title{Universal Character, Phase Model and Topological Strings on $\C^3$}
\author{
\  \ Na Wang\dag, Chuanzhong Li\ddag\footnote{Corresponding author:lichuanzhong@nbu.edu.cn }\\
\dag\small Department of mathematics and statistics, Henan University, Kaifeng, 475001, China.\\
\ddag\small Department of Mathematics,  Ningbo University, Ningbo, 315211, China}

\date{}
\maketitle

\begin{abstract}
In this paper, we consider two different subjects: the algebra of universal characters $S_{[\lambda,\mu]}({\bf x},{\bf y})$ (a generalization of Schur functions) and the phase model of strongly correlated bosons. We find that the two-site generalized phase model can be realized in the algebra of universal characters, and the entries in the monodromy matrix of the phase model  can be represented by the vertex operators $\Gamma_i^\pm(z) (i=1,2)$ which generate universal characters. Meanwhile, we find that these vertex operators can also be used to obtain the A-model topological string partition function on $\C^3$.

\end{abstract}
\noindent
{\bf Keywords: }{universal character, vertex operator, phase model, topological strings.}

\section{Introduction}\label{sec1}
Symmetric functions were used to determine irreducible characters of highest weight representations of the classical groups\cite{weyl}. The universal character, as a generalization of Schur function, describes the character of an irreducible rational representation of $GL(n)$\cite{KK}, which upgrades that Schur function is the character of an irreducible polynomial representation of $GL(n)$. Symmetric functions also appear in mathematical physics, especially in integrable models. The group in the  Kyoto school uses Schur functions in a remarkable way to understand the KP and KdV hierarchies\cite{MJD}. T. Tsuda defined the UC hierarchy which is a generalization of KP hierarchy and obtained that the tau functions of UC hierarchy can be realized in terms of the universal characters. He also proved that the UC hierarchy has relations with Painlev\'e equations\cite{Tsuda} by similar reductions. In this paper, we consider two different subjects: the algebra of universal characters and the phase model of strongly correlated bosons.

One purpose of this paper is to give the representation of the two-site generalized phase model on the algebra of universal characters $S_{[\lambda,\mu]}({\bf x},{\bf y})$.  The phase model, which is the so-called crystal limit of the quantum group \cite{MK}, is an integrable model and can be solved in the formulism of the quantum inverse scattering method \cite{KBI}. Our results will show that the limit of the quantum inverse scattering method has an interpretation in terms of the algebra of universal characters. The crucial elements in our discussion are vertex operators $\Gamma_i^\pm(z),\ (i=1,2)$. The Fermions can be defined from these vertex operators. In the special case $\mu=\emptyset$, the universal character $S_{[\lambda,\mu]}({\bf x},{\bf y})$ will be reduced to the Schur function $S_{\lambda}({\bf x})$, and the correspondence between vertex operators and fermions in this special case is a part of the well-known Boson-Fermion correspondence.

The relation between the algebra of Schur functions and the phase model is known from \cite{PS,NVT}. There is the following isometry between states in the phase model and Schur functions
\[
\bigotimes_{i=0}^{M} |n_i\rangle_i\mapsto S_{\lambda}({\bf x}),\quad \lambda=1^{n_1}2^{n_2}\ldots.
\]
The actions of the entries in monodromy matrix $T(u)$ on Schur function are obtained from the truncated expansions of the vertex operators
\[
\tilde{\Gamma}^+(z)=e^{\xi({\bf x},z)},\quad \tilde{\Gamma}^-(z)=e^{\xi(\tilde{\partial}_{\bf x},z^{-1})}
\]
where $\xi({\bf x},z)=\sum_{n=1}^\infty x_n z^n$ and
$$\tilde{\partial}_{\bf x}=(\partial_{x_1}, \frac{1}{2}\partial_{x_2},\frac{1}{3}\partial_{x_3},\cdots), \quad \partial_{x_i}=\frac{\partial}{\partial_{x_i}}.$$ In this paper, we generalize vertex operators to the type
\bea
&&\Gamma_1^-(z)=e^{\xi({\bf x}-\tilde{\partial}_{\bf y}, z)},\quad \Gamma_1^+(z)=e^{\xi(\tilde{\partial}_{\bf x}, z^{-1})},\nonumber\\
&&\Gamma_2^-(z)=e^{\xi({\bf y}-\tilde{\partial}_{\bf x}, z)},\quad \Gamma_2^+(z)=e^{\xi(\tilde{\partial}_{\bf y}, z^{-1})},\nonumber
\eea
these vertex operators can generate universal characters and fermions can be defined from them. We define the map from the state in the phase model to the universal character by
\[
\bigotimes_{i=0}^{M_1} |n_i\rangle_i^{(1)}\bigotimes\bigotimes_{i=0}^{M_2} |m_i\rangle_i^{(2)}\mapsto S_{[\lambda,\mu]}({\bf x},{\bf y})
\]
with
\[
\lambda=1^{n_1}2^{n_2}\ldots,\quad \mu=1^{m_1}2^{m_2}\ldots.
\]
The actions of creation operators $B_1(u)$ and $B_2(u)$, which are entries in the monodromy matrix $T(u)$ of the phase model, on universal characters are determined by the truncated expansions of $\Gamma_1^-(z)$ and $\Gamma_2^-(z)$ respectively, and the annihilation operators $C_1(u)$ and $C_2(u)$ (multiplied by a coefficient) are adjoint to the operators $B_1(u)$ and $B_2(u)$, and determined by the truncated expansions of $\Gamma_1^+(z)$ and $\Gamma_2^+(z)$ respectively.

Another purpose of this paper is to discuss the relations between the vertex operators $\Gamma_i^-(z)$, $\Gamma_i^+(z)$ $(i=1,2)$ and the MacMahon functions. It is known that the A-model topological string partition function $Z_{\C^{3}}^{top}$ on $\C^3$ can be written as a Fermionic correlator involving the vertex operators $\tilde{\Gamma}^+(z)$ and $\tilde{\Gamma}^-(z)$ with a particular specialization of the values of $z=q^{\pm 1/2},q^{\pm 3/2},q^{\pm 5/2}\cdots$. In this paper, we will give that $Z_{\C^{3}}^{top}$ can also be obtained from the vertex operators $\Gamma_i^-(z)$ and $\Gamma_i^+(z)$ $(i=1,2)$  with the same specialization of the values of $z$.

The paper is organized as follows. In section \ref{sect2}, we recall the definition of universal character and its vertex operator realization, then we give the actions of the vertex operators on $S_{[\lambda,\mu]}({\bf x},{\bf y})$ which is helpful for our discussion. In section \ref{sect3}, we recall the phase model. In section \ref{sect4}, we define the representation of the two-site generalized phase model on the algebra of universal characters, and we find that the actions of the entries in monodromy matrix on universal characters are obtained from the truncated expansions of the vertex operators discussed in section \ref{sect2}. In section \ref{sect5}, we give that the MacMahon function can be obtained from these vertex operators.
\section{Universal characters and vertex operators}\label{sect2}
Let ${\bf x}=(x_1,x_2,\cdots)$ and ${\bf y}=(y_1,y_2,\cdots)$.
The operators $h_n({\bf x})$ are determined by the generating function:
\be\label{hxi}
\sum_{n=0}^\infty h_n({\bf x})z^n=e^{\xi({\bf x},z)},\quad \xi({\bf x},z)=\sum_{n=1}^\infty x_n z^n
\ee
and set $h_n({\bf x})=0$ for $n<0$. The operators $h_n({\bf x})$ can be explicitly written as
\[
h_n({\bf x})=\sum_{k_1+2k_2+\cdots nk_n=n}\frac{x_1^{k_1}x_2^{k_2}\cdots x_n^{k_n}}{k_1!k_2!\cdots k_n!}.
\]
Note that $h_n({\bf x})$ is the complete homogeneous symmetric function if we replace $ix_i$ with the power sum $p_i$.

For a pair of Young diagrams $\lambda=(\lambda_1,\lambda_2,\cdots,\lambda_l)$ and $\mu=(\mu_1,\mu_2,\cdots,\mu_{l'})$, the universal character $S_{[\lambda,\mu]}=S_{[\lambda,\mu]}({\bf x},{\bf y})$ is a polynomial of variables ${\bf x}$ and ${\bf y}$ in $\C[{\bf x},{\bf y}]$ defined by the twisted Jacobi-Trudi formula \cite{KK}:
\be
S_{[\lambda,\mu]}({\bf x},{\bf y})=\text{det}\left( \begin{array}{cc}
h_{\mu_{l'-i+1}+i-j}({\bf y}), & 1\leq i\leq l' \\
h_{\lambda_{i-l'}-i+j}({\bf x}), & l'+1\leq i\leq l+l'
\end{array} \right)_{1\leq i,j\leq l+l'}.
\ee
Define the degree of each variables $x_n, y_n, \ n=1,2,\cdots$ by
\[
\text{deg }x_n=n,\quad \text{deg }y_n=-n
\]
then $S_{[\lambda,\mu]}({\bf x},{\bf y})$ is a homogeneous polynomial of degree $|\lambda|-|\mu|$, where $|\lambda|=\lambda_1+\lambda_2+\cdots+\lambda_l$ is called the weight of $\lambda$. Note that $S_\lambda({\bf x})$ is a special case of the universal character: $S_\lambda({\bf x})=\det(h_{\lambda_i-i+j}({\bf x}))=S_{[\lambda,\emptyset]}({\bf x},{\bf y})$.

Introduce the following vertex operators
\bea
&&\Gamma_1^-(z)=e^{\xi({\bf x}-\tilde{\partial}_{\bf y}, z)},\quad \Gamma_1^+(z)=e^{\xi(\tilde{\partial}_{\bf x}, z^{-1})},\label{g1}\\
&&\Gamma_2^-(z)=e^{\xi({\bf y}-\tilde{\partial}_{\bf x}, z)},\quad \Gamma_2^+(z)=e^{\xi(\tilde{\partial}_{\bf y}, z^{-1})}.\label{g2}
\eea
Define
\bea
X^\pm(z)&=&\sum_{n\in\Z}X^\pm_nz^n=e^{\pm\xi({\bf x}-\tilde{\partial}_{\bf y}, z)}e^{\mp\xi(\tilde{\partial}_{\bf x}, z^{-1})},\\
Y^\pm(z^{-1})&=&\sum_{n\in\Z}Y^\pm_nz^{-n}=e^{\pm\xi({\bf y}-\tilde{\partial}_{\bf x}, z^{-1})}e^{\mp\xi(\tilde{\partial}_{\bf y}, z)}.
\eea
The operators $X_i^\pm$ satisfy the following Fermionic relations:
\bea
X_i^\pm X_j^\pm +X_{j-1}^\pm X_{i+1}^{\pm}&=&0,\nonumber\\
X_i^+X_j^-+X_{j+1}^{-}X_{i-1}^+&=&\delta_{i+j,0}.\nonumber
\eea
The same relations hold also for $Y_i^\pm$, and $X_i^\pm$ and $Y_i^\pm$ are commutative. The operators $X_i^+$ and $Y_i^+$ are raising operators for the universal characters such that
\begin{equation}
S_{[\lambda,\mu]}({\bf x},{\bf y})=X_{\lambda_1}^+\cdots X_{\lambda_l}^+Y_{\mu_1}^+\cdots X_{\mu_{l'}}^+\cdot 1
\end{equation}
where the Young diagrams $\lambda=(\lambda_1,\lambda_2,\cdots,\lambda_l)$ and $\mu=(\mu_1,\mu_2,\cdots,\mu_{l'})$.

For an unknown function $\tau=\tau({\bf x},{\bf y})$, the bilinear relations\cite{Tsuda}
\be
\sum_{i+j=-1}X_i^-\tau\otimes X_j^+\tau=\sum_{i+j=-1}Y_i^-\tau\otimes Y_j^+\tau=0
\ee
is called the UC hierarchy.

It turns out that $\tau({\bf x},{\bf y})$ equals
\[
\tau({\bf x},{\bf y})=\tau_1({\bf x}-\tilde{\partial}_{\bf y})\tau_2({\bf y}-\tilde{\partial}_{\bf x})\cdot 1
\]
where $\tau_1({\bf x})$ and $\tau_2({\bf x})$ are tau functions of KP hierarchy.
In special case, the universal character $S_{[\lambda,\mu]}({\bf x},{\bf y})$ is also the solution of UC hierarchy, and
\be
S_{[\lambda,\mu]}({\bf x},{\bf y})=S_\lambda({\bf x}-\tilde{\partial}_{\bf y})S_\mu({\bf y}-\tilde{\partial}_{\bf x})\cdot 1.
\ee
From this, we get the conclusion which is helpful for the following discussion.
\begin{prp}
The vertex operators in (\ref{g1}) and (\ref{g2}) act on the universal characters $S_{[\lambda,\mu]}({\bf x},{\bf y})$ in the following way,
\bea
\Gamma_1^-(z)S_{[\lambda,\mu]}({\bf x},{\bf y})&=&\sum_{n=0}^\infty z^n S_{[\lambda\cdot (n),\mu]}({\bf x},{\bf y})=\sum_{\nu\succ\lambda}z^{|\nu|-|\lambda|} S_{[\nu,\mu]}({\bf x},{\bf y}),\\
\Gamma_2^-(z)S_{[\lambda,\mu]}({\bf x},{\bf y})&=&\sum_{n=0}^\infty z^n S_{[\lambda,\mu\cdot (n)]}({\bf x},{\bf y})=\sum_{\nu\succ\mu}z^{|\nu|-|\mu|} S_{[\lambda,\nu]}({\bf x},{\bf y}),
\eea
where the multiplication of two Young diagram $\lambda$ and $(n)$ satisfies the Pieri formula\cite{Mac,FH}, and $\lambda\succ\mu$ means that the Young diagrams $\lambda$ and $\mu$ are interlaced, in the sense of $\lambda_1\geq\mu_1\geq\lambda_2\geq\mu_2\geq\cdots$. The operators $\Gamma_1^+(z)$ and $\Gamma_2^+(z)$ are adjoint to $\Gamma_1^-(z)$ and $\Gamma_2^-(z)$ respectively, that is,
\bea
\Gamma_1^+(z)S_{[\lambda,\mu]}({\bf x},{\bf y})&=&\sum_{n=0}^\infty z^n S_{[\lambda/ (n),\mu]}({\bf x},{\bf y})=\sum_{\lambda\succ\nu}z^{|\lambda|-|\nu|} S_{[\nu,\mu]}({\bf x},{\bf y}),\\
\Gamma_2^-(z)S_{[\lambda,\mu]}({\bf x},{\bf y})&=&\sum_{n=0}^\infty z^n S_{[\lambda,\mu/(n)]}({\bf x},{\bf y})=\sum_{\mu\succ\nu}z^{|\mu|-|\nu|} S_{[\lambda,\nu]}({\bf x},{\bf y}),
\eea
where $\lambda/\mu$ denotes a skew diagram.
\end{prp}
\section{Phase model}\label{sect3}
We begin this section with a bosonic system based on the following algebra\cite{PS,NVT}
\begin{equation}\label{nphi}
[N,\phi] = - \phi,\quad [N,\phi^{\dag}] = \phi^{\dag},\quad [\phi,\phi^{\dag}] = \pi,
\end{equation}
where $\pi=|0\rangle \langle 0|$ is the vacuum projection. The operator $\phi$ is one-sided isometry
$$
\phi \phi^{\dag} = 1,\qquad \phi^{\dag} \phi = 1-\pi.
$$
This algebra can be represented in the Fock space $\mathcal{F}$ consisting of $n$-particle states
$|n\rangle$, the operators $\phi$, $\phi^\dag$ and $N$ acting as the phase operators and the number of particles operator, respectively,
$$
\phi^{\dag}|n\rangle =|n+1\rangle,\quad \phi|n\rangle = |n-1\rangle, \quad \phi|0\rangle = 0, \quad
N|n\rangle = n|n\rangle,
$$
where $|0\rangle$ is the vacuum state, the special case $n=0$ of the $n$ particle state.

Let the tensor product
\be
\mathcal{F} = \bigotimes_{i=0}^M \mathcal{F}_i,
\ee
be $M+1$ copies of the Fock space. Denote by $\phi_i, \phi^{\dag}_i, N_i$ the operators that act as $\phi, \phi^{\dag}, N$ in (\ref{nphi}), respectively, in the $i$th space and identically in the other spaces.

The phase model is a model of a periodic chain with the hamiltonian \cite{BKT,BIK,bogo}
\be
H = -\frac{1}{2} \sum_{i=0}^M \big(\phi^{\dag}_i \phi_{i+1} + \phi_i \phi^{\dag}_{i+1} -2N_i \big).
\ee
Define the operator of the total number of particles by
$$
\hat{N} = \sum_{i=0}^M  N_i.
$$
Then the $N$-particle vectors in this space are of the form
\be
\bigotimes_{i=0}^M |n_i\rangle_i,\qquad \textrm{where} \ |n_i\rangle_i =
(\phi_i^{\dag})^{n_i}|0\rangle_i,\quad  N=\sum_{i=0}^M n_i,    \label{Npar}
\ee
the numbers $n_i$ are called the occupation numbers of the state (\ref{Npar}).

From \cite{KBI}, we know that the phase model is integrable. Introduce the $L$-matrix
$$
L_i(u) = \left( \begin{array}{cc}
u^{-1} & \phi^{\dag}_i \\
\phi_i & u
\end{array} \right), \qquad i=0,\ldots, M,
$$
where $u$ is a scalar parameter, here we treat $u$ as $uI$ with $I$ being the identity operator in $\mathcal F$. For every $i=0,\ldots, M$, the $L$-matrix satisfies the bilinear equation
\be
R(u,v)\big(L_i(u)\otimes L_i(v)  \big)  =  \big(L_i(v)\otimes L_i(u)  \big) R(u,v),
\ee
where $R$-matrix $R(u,v)$ is a $4\times4$ matrix given by
\be
R(u,v) = \left( \begin{array}{cccc}
f(v,u) & 0 & 0 & 0 \\
0 & g(v,u) & 1 & 0 \\
0 & 0 & g(v,u) & 0 \\
0 & 0 & 0 & f(v,u)
\end{array} \right),   \label{Rmatrix}
\ee
with
$$
f(v,u) = \frac{u^2}{u^2-v^2},\quad g(v,u) = \frac{uv}{u^2-v^2}.
$$

Define the monodromy matrix by
$$
T(u) = L_M(u) L_{M-1}(u)\cdots L_0(u),
$$
which gives the solution of the phase model. It also satisfies the bilinear equation
\be\label{rtt}
R(u,v)\big(T(u)\otimes T(v)  \big)  = \big(T(v)\otimes T(u)  \big) R(u,v).
\ee

Let
\be
T(u) =  \left( \begin{array}{cc}
A(u) & B(u) \\
C(u) & D(u)
\end{array} \right),    \label{ABCD}
\ee
we have
\be
\hat{N} B(u) = B(u) (\hat{N}+1), \qquad \hat{N} C(u) = C(u) (\hat{N}-1).   \label{creat-annih}
\ee
Therefore, we call $B(u)$ the creation operator and $C(u)$ the annihilation operator. The  operators $A(u)$ and $D(u)$ do not change the total number of particles.

Denote by $|0\rangle_j$ the vacuum vector in $\mathcal F_j$ and by $|0\rangle=\otimes_{i=0}^M |0\rangle_i$. Let
\be
|\Psi(u_1,\ldots,u_N) \rangle = \prod_{i=1}^N B(u_j) |0\rangle,   \label{Psi}
\ee
which is a $N$ particle state.

According to \cite{PS,NVT}, there is the following isometry between the states (\ref{Npar}) and the Schur functions
\begin{equation}
\bigotimes_{i=0}^{M} |n_i\rangle_i\mapsto S_{\lambda}({\bf x}),\quad \lambda=1^{n_1}2^{n_2}\ldots,
\end{equation}
and the operator $B(u)$ acts on Schur functions as the operator of multiplication by $u^MH_M(u^2)$, where $H_M(t)=\sum_{k=0}^Mt^kh_k$ is the truncated generating function of the complete homogeneous symmetric functions $h_k$. Then the state $|\Psi(u_1,\ldots,u_N) \rangle$ has the following expansion
\begin{equation}
|\Psi(u_1,\ldots,u_N) \rangle=\sum_\lambda S_\lambda(u_1^2,\cdots,u_N^2)\bigotimes_{i=0}^M |n_i\rangle_i,\quad \lambda=1^{n_1}2^{n_2}\ldots.
\end{equation}

In the following, we will generalize this work \cite{PS,NVT} to realize the phase model in the algebra of universal characters, and the entries in the monodromy matrix can be obtained from the vertex operators which generate the universal characters.
\section{Two-site generalized phase model and universal characters}\label{sect4}
Now fix two positive integers $M_1,M_2$ and consider the tensor products
\be
\mathcal{F}^{(1)} = \bigotimes_{i=0}^M \mathcal{F}_i^{(1)}\quad\quad \mathcal{F}^{(2)} = \bigotimes_{i=0}^M \mathcal{F}_i^{(2)},
\ee
which $\mathcal{F}^{(1)}$ and $\mathcal{F}^{(2)}$ are $M_1+1$ and $M_2+1$ copies of the Fock space respectively. Denote by $\phi_i^{(1)}, \phi^{(1)\dag}_i, N_i^{(1)}$ the operators that act as $\phi, \phi^{\dag}, N$ in (\ref{nphi}), respectively, in the $i$th space $\mathcal{F}_i^{(1)}$ and identically in the other spaces of $\mathcal{F}^{(1)}$ and in all spaces of $\mathcal{F}^{(2)}$, and denote by $\phi_i^{(2)}, \phi^{(2)\dag}_i, N_i^{(2)}$ the operators that act as $\phi, \phi^{\dag}, N$, respectively, in the $i$th space $\mathcal{F}_i^{(2)}$ and identically in the other spaces of $\mathcal{F}^{(2)}$ and in all spaces of $\mathcal{F}^{(1)}$.
Let
\be
\mathcal{F}=\mathcal{F}^{(1)}\bigotimes\mathcal{F}^{(2)}
\ee
so that $\mathcal{F}$ is $M_1+M_2+2$ copies of the Fock space. Denote by $|0\rangle_j^{(i)}$ the vacuum vector in $\mathcal F_j^{(i)}$ with $i=1,2.$

Define the operators
$$
\hat{N}_1 = \sum_{i=0}^M  N_i^{(1)}, \quad \hat{N}_2 = \sum_{i=0}^M  N_i^{(2)},\quad \text{and} \quad \hat N=\hat{N}_1+\hat{N}_2.
$$
The $(N_1,N_2)$-particle vectors in space $\mathcal{F}$ are of the form
\be\label{nmbasis}
\bigotimes_{i=0}^{M_1} |n_i\rangle_i^{(1)}\bigotimes\bigotimes_{i=0}^{M_2} |m_i\rangle_i^{(2)},\qquad \textrm{with}\quad N_1=\sum_{i=0}^{M_1} n_i,\quad N_2=\sum_{i=0}^{M_2} m_i,
\ee
where
\[
|n_i\rangle_i^{(1)}= (\phi_i^{(1)\dag})^{(n_i)}|0\rangle_j^{(1)},\quad |m_i\rangle_i^{(2)}= (\phi_i^{(2)\dag})^{(m_i)}|0\rangle_j^{(2)}.
\]

Define the map
$\jmath: \mathcal F\rightarrow \C[{\bf x}, {\bf y}]$
by
\be\label{jmath}
\jmath(\bigotimes_{i=0}^{M_1} |n_i\rangle_i^{(1)}\bigotimes\bigotimes_{i=0}^{M_2} |m_i\rangle_i^{(2)})= S_{[\lambda,\mu]}({\bf x},{\bf y})
\ee
with
\be
\lambda=1^{n_1}2^{n_2}\ldots,\quad \mu=1^{m_1}2^{m_2}\ldots.
\ee
In fact, this association is not quite unique: partitions $\lambda,\ \mu$ themselves do not know about numbers $n_0$ and $m_0$ of particles. Nonetheless, if we fix the total numbers of particles $N_1$ and $N_2$, we can deduce $n_0=N_1-l(\lambda)$ and $m_0=N_2-l(\mu)$, where $l(\lambda)$ is the length of partition $\lambda$, that is, the number of rows in $\lambda$.
Note that the correspondence (7) in \cite{NVT} is a special case of the map $\jmath$ defined above.

The $L$-matrices are
$$
L^{(1)}_i(u) = \left( \begin{array}{cc}
u^{-1} & \phi^{(1)\dag}_i \\
\phi_i^{(1)} & u
\end{array} \right), \qquad i=0,\ldots, M_1,
$$
$$
L^{(2)}_i(u) = \left( \begin{array}{cc}
u^{-1} & \phi^{(2)\dag}_i \\
\phi_i^{(2)} & u
\end{array} \right), \qquad i=0,\ldots, M_2,
$$
and the monodromy matrix is
\bea
T(u)&=&L^{(2)}_{M_2}(u) \cdots L^{(2)}_0(u)L^{(1)}_{M_1}(u) \cdots L^{(1)}_0(u).\nonumber
\eea
Each $L$-matrix, as well as the monodromy matrix, satisfies the bilinear equation again
\bea
R(u,v)\big(L^{(j)}_i(u)\otimes L^{(j)}_i(v)  \big) & = & \big(L^{(j)}_i(v)\otimes L^{(j)}_i(u)  \big) R(u,v),\quad j=1,2,\nonumber \\
R(u,v)\big(T(u)\otimes T(v)  \big) & = & \big(T(v)\otimes T(u)  \big) R(u,v), \label{intertwine}
\eea
with the same $R$-matrix in (\ref{Rmatrix}).

Let
\bea
T_i(u)&=&L^{(i)}_{M_i}(u)L^{(i)}_{M_i-1}(u) \cdots L^{(1)}_0(u)\\
&=&\left( \begin{array}{cc}
A_i(u) & B_i(u) \\
C_i(u) & D_i(u)
\end{array} \right),\quad i=1,2.
\eea
The operators $B_1(u)$ and $B_2(u)$ are called the creation operators and $C_1(u)$ and $C_2(u)$ the annihilation operators, in the sense that they increase and decrease the total numbers of particles
\be
\hat{N}_i B_i(u) = B_i(u) (\hat{N}_i+1), \qquad \hat{N}_i C_i(u) = C_i(u) (\hat{N}_i-1) \quad\text{for}\quad i=1,2.
\ee
The operators $A_i(u)$ and $D_i(u)\ (i=1,2)$   do not change the total number of particles.

We call $i$th Fock space $\mathcal F_i^{(j)}\ (j=1,2)$ the $i$-energy space. Note that the correspondence (\ref{jmath}) does not take into account the numbers $n_0,\ m_0$ of zero-energy particles, Therefore (\ref{jmath}) gives a representation of the positive-energy space
\[
\hat {\mathcal F}=\bigotimes_{i=1}^M \mathcal{F}_i^{(1)}\bigotimes \bigotimes_{i=1}^M \mathcal{F}_i^{(2)},
\]
in the algebra of symmetric functions $\C[{\bf x},{\bf y}]$, in fact in its subspace $\C_{M_1,M_2}[{\bf x},{\bf y}]$ generated by universal characters $S_{[\lambda,\mu]}$ where the Young diagrams $\lambda$ have at most $M_1$ columns and $\mu$ at most $M_2$ columns. From the definition of the twisted Jacobi-Trudi formula, we can define this subspace by supposing
\[
h_{M_1+1}({\bf x})=h_{M_1+2}({\bf x})=\cdots=0,\quad h_{M_2+1}({\bf y})=h_{M_2+2}({\bf y})=\cdots=0.
\]

By the definition of $\hat N_1,\ \hat N_2$, the space $\mathcal F$ has a decomposition into $(N_1,N_2)$-particle subspaces $\mathcal F^{N_1,N_2}$, i.e.,
\be
\mathcal F=\bigoplus_{N_1,N_2\geq 0}\mathcal F^{N_1,N_2}=\bigoplus_{N_1,N_2\geq 0}\mathcal F_1^{N_1}\otimes\mathcal F_2^{N_2}.
\ee
Under the map (\ref{jmath}), the subspace $\mathcal F^{N_1,N_2}$ corresponds to $\C_{M_1,M_2}^{N_1,N_2}[{\bf x},{\bf y}]$ which is spanned by universal characters $S_{[\lambda,\mu]}$ whose diagrams $\lambda$ lie in the $N_1\times M_1$ box and $\mu$ lie in the $N_2\times M_2$ box. That Young diagram $\lambda$ lies in the $N\times M$ box means $\lambda$ has at most $N$ rows and at most $M$ columns.

Define the projection $P:\mathcal F\rightarrow \hat{\mathcal F}$ by forgetting the zero energy states, and define the operator $\mathcal{B}_i(u):=PB_i(u)P$ for $i=1,2$. Then $\mathcal{B}_i(u)$ are operators acting on the space $\hat{\mathcal F}\cong\C_{M_1,M_2}[{\bf x},{\bf y}]$. Since $B_i(u)$ are creation operators, then $\mathcal{B}_1(u)$ sends $\C_{M_1,M_2}^{N_1,N_2}[{\bf x},{\bf y}]$ to $\C_{M_1,M_2}^{N_1+1,N_2}[{\bf x},{\bf y}]$ and $\mathcal{B}_2(u)$ sends $\C_{M_1,M_2}^{N_1,N_2}[{\bf x},{\bf y}]$ to $\C_{M_1,M_2}^{N_1,N_2+1}[{\bf x},{\bf y}]$. In the following, we discuss the actions of $\mathcal{B}_i(u)$ on $\C_{M_1,M_2}[{\bf x},{\bf y}]$.
Define $\tilde{\mathcal{B}}_i(u)$ by $\mathcal{B}_i(u)=u^{-M_i}\tilde{\mathcal{B}}_i(u)$ for $i=1,2$.
Then we can derive the following proposition.

\begin{prp}\label{prpbb}
In the space $\C_{M_1,M_2}[{\bf x},{\bf y}]$,
\bea
\tilde{\mathcal{B}}_1(u)&=& H_{M_1}({\bf x}-\tilde{\partial}_{\bf y},u^2),\\
\tilde{\mathcal{B}}_2(u)&=& H_{M_2}({\bf y}-\tilde{\partial}_{\bf x},u^2),
\eea
where $H_n({\bf x},t)=\sum_{k=0}^n t^kh_k({\bf x})$, and $h_k({\bf x})$ is defined in (\ref{hxi}),  which in fact is the complete homogeneous symmetric function.
\end{prp}

To prove this proposition, we need the following lemma.
\begin{lem} For any Schur functions $S_\lambda$ and complete symmetric function $h_k$, we have
\be
h_k({\bf y}-\tilde{\partial}_{\bf x})S_\lambda({\bf x}-\tilde{\partial}_{\bf y})=S_\lambda({\bf x}-\tilde{\partial}_{\bf y})h_k({\bf y}-\tilde{\partial}_{\bf x})
\ee
\end{lem}
\begin{proof}
this holds since $y_n-\frac{1}{n}\partial_{x_n}$ and $x_m-\frac{1}{m}\partial_{y_m}$ are commutative.
\end{proof}
{\it The proof of Proposition \ref{prpbb}}.
\begin{proof}
We know that
\[
S_{[\lambda,\mu]}({\bf x},{\bf y})=S_{\lambda}({\bf x}-\tilde{\partial}_{\bf y})S_\mu({\bf y}-\tilde{\partial}_{\bf x})\cdot 1,
\]
and in \cite{PS,NVT}
\[
\tilde{\mathcal B}(u)\bigotimes_{i=0}^M |n_i\rangle_i=\sum_{k=0}^Mu^{2k}h_k S_\lambda, \quad \lambda=1^{n_1}2^{n_2}\ldots.
\]
Here the mapping sign $\jmath$ is omitted at the left of the equation and we will do the same in the following.
Hence,
\bea
\tilde{\mathcal B}_1(u)\bigotimes_{i=0}^{M_1} |n_i\rangle_i^{(1)}\bigotimes\bigotimes_{i=0}^{M_2} |m_i\rangle_i^{(2)}&=&\big(\tilde{\mathcal B}_1(u)\bigotimes_{i=0}^{M_1} |n_i\rangle_i^{(1)}\big)\bigotimes\bigotimes_{i=0}^{M_2} |m_i\rangle_i^{(2)}\nonumber\\
&=& \sum_{k=0}^{M_1}u^{2k}h_k({\bf x}-\tilde{\partial}_{\bf y}) S_\lambda({\bf x}-\tilde{\partial}_{\bf y})S_\mu({\bf y}-\tilde{\partial}_{\bf x})\cdot 1,\nonumber
\eea
and
\bea
\tilde{\mathcal B}_2(u)\bigotimes_{i=0}^{M_1} |n_i\rangle_i^{(1)}\bigotimes\bigotimes_{i=0}^{M_2} |m_i\rangle_i^{(2)}&=&\bigotimes_{i=0}^{M_1} |n_i\rangle_i^{(1)}\bigotimes\big(\tilde{\mathcal B}_2(u)\bigotimes_{i=0}^{M_2} |m_i\rangle_i^{(2)}\big)\nonumber\\
&=&S_\lambda({\bf x}-\tilde{\partial}_{\bf y})\big(\sum_{k=0}^{M_2}u^{2k}h_k({\bf y}-\tilde{\partial}_{\bf x})S_\mu({\bf y}-\tilde{\partial}_{\bf x})\big)\cdot 1\nonumber\\
&=& \sum_{k=0}^{M_2}u^{2k}h_k({\bf y}-\tilde{\partial}_{\bf x}) S_\lambda({\bf x}-\tilde{\partial}_{\bf y})S_\mu({\bf y}-\tilde{\partial}_{\bf x})\cdot 1.\nonumber
\eea
\end{proof}

Note that the truncated generating function $H_n({\bf x},t)=\sum_{k=0}^n t^kh_k({\bf x})$ can also be regarded as the full generating function $H({\bf x},t)=\sum_{k=0}^\infty t^kh_k({\bf x})$ under the specialization: $H_n({\bf x},t)=H({\bf x},t)|_{h_{n+1}({\bf x})=h_{n+2}({\bf x})=\ldots=0}$. Then we get
the following corollary.
\begin{cor}
In the $M_1,M_2\rightarrow \infty$ limit, the actions of the creation operators $\tilde{\mathcal B}_1(u)$ and $\tilde{\mathcal B}_2(u)$ on the space $\hat{\mathcal F}$ correspond to the multiplications of $H({\bf x}-\tilde{\partial}_{\bf y},u^2)$ and $H({\bf y}-\tilde{\partial}_{\bf x},u^2)$ on the space $\C[{\bf x},{\bf y}]$ respectively.
\end{cor}

Define
\be\label{tildepsi}
|\tilde\Psi_N(u_1,\cdots,u_N)\rangle:=\prod_{j=1}^NB_2(u_j)B_1(u_j)|0\rangle
\ee
then we obtain the following proposition.
\begin{prp}
The expansion of the $(N,N)$-particle vector (\ref{tildepsi}) in terms of basis vector (\ref{nmbasis}) is given by the formula
\bea
&&|\tilde\Psi_N(u_1,\cdots,u_N)\rangle\nonumber\\
&=&(u_1\cdots u_N)^{-M_1-M_2}\sum_{\lambda,\mu}S_\lambda(u_1^2,\cdots,u_N^2)S_\mu(u_1^2,\cdots,u_N^2)\bigotimes_{i=0}^{M_1} |n_i\rangle_i^{(1)}\bigotimes\bigotimes_{i=0}^{M_2} |m_i\rangle_i^{(2)}\nonumber\\
&=&(u_1\cdots u_N)^{-M_1-M_2}\sum_{\lambda,\mu}S_\lambda(u_1^2,\cdots,u_N^2)S_\mu(u_1^2,\cdots,u_N^2)S_{[\lambda,\mu]}({\bf x},{\bf y})\nonumber
\eea
where the sum is over Young diagrams $\lambda$ with at most $N$ rows and at most $M_1$ columns, Young diagrams $\mu$ with at most $N$ rows and at most $M_2$ columns.
\end{prp}
\begin{proof}
The operators $B_1(u)$ and $B_2(u)$ are commutative.
\bea
\prod_{j=1}^NB_1(u_j)&=&(u_1\cdots u_N)^{-M_1}\sum_{\lambda}S_{\lambda}(u_1^2,\cdots,u_N^2)S_\lambda({\bf x}-\tilde{ \partial}_{\bf y})\\
\prod_{j=1}^NB_2(u_j)&=&(u_1\cdots u_N)^{-M_2}\sum_{\mu}S_{\mu}(u_1^2,\cdots,u_N^2)S_\mu({\bf y}-\tilde{ \partial}_{\bf x})
\eea
then we have
\bea
&&|\tilde\Psi_N(u_1,\cdots,u_N)\rangle\nonumber\\
&=&(u_1\cdots u_N)^{-M_1-M_2}\sum_{\lambda,\mu}S_\lambda(u_1^2,\cdots,u_N^2)S_\mu(u_1^2,\cdots,u_N^2)S_\lambda({\bf x}-\tilde{ \partial}_{\bf y})S_\mu({\bf y}-\tilde{ \partial}_{\bf x})\cdot1\nonumber\\
&=& (u_1\cdots u_N)^{-M_1-M_2}\sum_{\lambda,\mu}S_\lambda(u_1^2,\cdots,u_N^2)S_\mu(u_1^2,\cdots,u_N^2)S_{[\lambda,\mu]}({\bf x},{\bf y}).
\eea
By the restrictions on $\lambda$ and $\mu$, we obtain the conclusion.
\end{proof}

Recall that
\bea
&&T(u)=\left( \begin{array}{cc}
A(u) & B(u) \\
C(u) & D(u)
\end{array} \right)\nonumber\\
&=&T_2(u)T_1(u)=\left( \begin{array}{cc}
A_2(u) & B_2(u) \\
C_2(u) & D_2(u)
\end{array} \right)\left( \begin{array}{cc}
A_1(u) & B_1(u) \\
C_1(u) & D_1(u)
\end{array} \right)\nonumber
\eea
then
\[
B(u)=A_2(u)B_1(u)+ B_2(u) D_1(u).
\]

The matrix entries of $T_i(u),i=1,2$ are related by the following formulas:
\[
B_i(u)=uA_i(u)\phi_0^{(i)\dag},\quad C_i(u)=u^{-1}\phi_0^{(i)}A^{\dag}_i(u^{-1}), \quad D_i(u)=\phi_0^{(i)}A^{\dag}_i(u^{-1})\phi_0^{(i)\dag}.
\]
Let $\mathcal{A}_i(u)=PA_i(u)P,\ \mathcal{C}_i(u)=PC_i(u)P,\ \mathcal{D}_i(u)=PD_i(u)P$, where $P$ is the projection $\mathcal F\rightarrow \hat{\mathcal F}$, then we have
\begin{lem}
The operators $\mathcal{A}_i(u),\ \mathcal{B}_i(u),\ \mathcal{C}_i(u),\ \mathcal{D}_i(u),\ (i=1,2)$ are related by the following formulas
\[
\mathcal{A}_i(u)=u^{-1}\mathcal{B}_i(u),\ \mathcal{C}_i(u)=\mathcal{B}^{\dag}_i(u^{-1}),\ \mathcal{D}_i(u)=u\mathcal{B}^{\dag}_i(u^{-1}).
\]
\end{lem}
Since $\mathcal{B}_1(u)=u^{-M_1}H_{M_1}({\bf x}-\tilde{ \partial}_{\bf y}, u^2)$, we obtain the following lemma.
\begin{lem}Let $\tilde{\mathcal{C}}_1(u)=u^{-M_1}\mathcal{C}_1(u)$, we have
\[
\tilde{\mathcal{C}}_1(u)= H_{M_1}^\bot({\bf x}-\tilde{\partial}_{\bf y},u^{-2})= H_{M_1}^\bot({\bf x},u^{-2})
\]
where $H_{M_1}^\bot({\bf x},t)=\sum_{k=0}^{M_1}t^kh_k^{\bot,M_1}({\bf x})$, and $h_k^{\bot,M_1}({\bf x})$ is the adjoint to the operator of multiplication by $h_k({\bf x})$.
\end{lem}
Define $\mathcal{B}(u)=PB(u)P$, we have
\bea
\mathcal{B}(u)&=&\mathcal{B}_2(u)(u^{-1}\mathcal{B}_1(u)+u\mathcal{B}_1^\dag(u^{-1}))\nonumber\\
&=&u^{-M_2}H_{M_2}({\bf y}-\tilde{\partial}_{\bf x},u^2)(u^{-M_1-1}H_{M_1}({\bf x}-\tilde{\partial}_{\bf y},u^2)+u^{M_1+1}H_{M_1}^\bot({\bf x}-\tilde{\partial}_{\bf y},u^{-2})).\nonumber
\eea
We write $H_{M_1}({\bf x}-\tilde{\partial}_{\bf y},u^2)$ by $H_1(u^2)$ for short. From the bilinear equation (\ref{rtt}), we have
\bea
u_1^{M_1+1}u_2^{-M_1-1}H_1^\bot(u_1^{-2})H_1(u_2^{2})&=&u_1^{M_1+1}u_2^{-M_1-1}\frac{u_1^2}{u_1^2-u_2^2}H_1(u_2^{2})H_1^\bot(u_1^{-2})\nonumber\\
&-&u_1^{-M_1-1}u_2^{M_1+1}\frac{u_1^2}{u_1^2-u_2^2}H_1(u_1^{2})H_1^\bot(u_2^{-2}).\nonumber
\eea
Recall that \[
|\Psi_N(u_1,\cdots,u_N)\rangle=\prod_{j=1}^N B(u_j)|0\rangle.
\]
 Then by calculation, we get
\begin{prp}
The operators $u_1^{-1}\mathcal{B}_1(u_1)+u_1\mathcal{B}_1^\dag(u_1^{-1})$ and $u_2^{-1}\mathcal{B}_1(u_2)+u_2\mathcal{B}_1^\dag(u_2^{-1})$ are commutative, which tells us that the coefficients, in the expansion $|\Psi_N(u_1,\cdots,u_N)\rangle$ in terms of basis vector (\ref{nmbasis}), are symmetric functions of variables $u_1^2,\cdots,u_N^2$.
\end{prp}

Since \bea
\prod_{j=1}^N \mathcal B(u_j)|0\rangle&=&\prod_{j=1}^N \mathcal{B}_2(u_j)(u_j^{-1}\mathcal{B}_1(u_j)+u_j\mathcal{B}_1^\dag(u_j^{-1}))|0\rangle\nonumber\\
&=&\prod_{j=1}^N \mathcal{B}_2(u_j)\prod_{j=1}^N(u_j^{-1}\mathcal{B}_1(u_j)+u_j\mathcal{B}_1^\dag(u_j^{-1}))|0\rangle\nonumber
\eea
and
\[
\prod_{j=1}^N \mathcal{B}_2(u_j)=(u_1\cdots u_N)^{-M_2}\sum_{\mu}S_\mu(u_1^2,\cdots,u_N^2)S_\mu({\bf y}-\tilde{ \partial}_{\bf x})
\]
where $\mu$ is a Young diagram with at most $N$ rows and at most $M_2$ columns.
Hence, in the following, we consider the expansion of $\prod_{j=1}^N(u_j^{-1}\mathcal{B}_1(u_j)+u_j\mathcal{B}_1^\dag(u_j^{-1}))|0\rangle$.

\begin{prp}
Let $k_1,k_2\cdots,k_i$ be in the set $\{1,2,\cdots, N\}$ and satisfy $k_1<k_2<\cdots<k_i$. We denote $u_{k_1}u_{k_2}\cdots u_{k_i}$ by $u_{\{k\}}$ for short. Then we have
\bea
&&\prod_{j=1}^N(u_j^{-1}\mathcal{B}_1(u_j)+u_j\mathcal{B}_1^\dag(u_j^{-1}))|0\rangle\nonumber\\
&=&(u_1\cdots u_N)^{M_1+1}\sum_{i=0}^N\sum_{k_1,\cdots,k_i}(u_{\{k\}})^{-2M_1-2}\prod_{j\neq k_i}\frac{u_j^2}{u_j^2-u_{k_i}^2}H_1(u_{k_1}^2)\cdots H_1(u_{k_i}^2)\cdot 1\nonumber
\eea
where
\[
H_1(u_{k_1}^2)\cdots H_1(u_{k_i}^2)\cdot 1=\sum_{\lambda}S_{\lambda}(u_{k_1}^2,\cdots,u_{k_i}^2)S_{\lambda}({\bf x}-\tilde{\partial}_{\bf y})\cdot 1.
\]
\end{prp}
\begin{proof}
One can prove it by inductions.
\end{proof}
From the discussion above, we get the expansion of $|\Psi_N(u_1,\cdots,u_N)\rangle$.
\begin{prp}The $N$-particle vector $|\Psi_N(u_1,\cdots,u_N)\rangle$ can be written as
\bea
|\Psi_N(u_1,\cdots,u_N)\rangle&=&(u_1\cdots u_N)^{-M_2+M_1+1}\sum_{i=0}^N\sum_{k_1,\cdots,k_i}(u_{\{k\}})^{-2M_1-2}\prod_{j\neq k_i}\frac{u_j^2}{u_j^2-u_{k_i}^2}\nonumber\\
&&\sum_{\lambda,\mu}S_{\lambda}(u_{k_1}^2,\cdots,u_{k_i}^2)S_\mu(u_1^2,\cdots,u_N^2)S_{[\lambda,\mu]}({\bf x},{\bf y})\nonumber
\eea
where $\lambda$ is a Young diagram with at most $i$ rows and $M_1$ columns, and $\mu$ a Young diagram with at most $N$ rows and $M_2$ columns.
\end{prp}
In a special case, we have $S_{[\lambda,\emptyset]}({\bf x},{\bf y})=S_{\lambda}({\bf x})$. Let $M_2=\emptyset$,
\bea
|\Psi_N(u_1,\cdots,u_N)\rangle&=&\prod_{j=1}^N B_1(u_j)|0\rangle=(u_1\cdots u_N)^{-M_1}\prod_{j=1}^N \tilde{\mathcal B}_1(u_j)|0\rangle\nonumber\\
&=&\sum_{\lambda}S_\lambda(u_1^2,\cdots,u_N^2)S_\lambda({\bf x})\nonumber
\eea
which is the same as in \cite{PS,NVT}.
\section{Vertex operators and topological strings on $\C^3$}\label{sect5}
Let $p_k$ denote the power sum, it is known that the following relation holds
\[
\sum_{k=0}^\infty h_k t^k=\exp{\sum_{k=1}^\infty \frac{p_k}{k}t^k}.
\]
Then we have the following proposition.
\begin{prp}
In the $M_1,M_2\rightarrow \infty$ limit, operators $\tilde{\mathcal B}_i(u)$ and $\tilde{\mathcal C}_i(u)$ have the following vertex operator representations
\bea
\tilde{\mathcal B}_1(u)&=&e^{\xi({\bf x}-\tilde{\partial}_{\bf y},u^2)}=\Gamma_1^+(u^2),\\
\tilde{\mathcal C}_1(u)&=&e^{\xi(\tilde{\partial}_{\bf x},u^{-2})}\ =\Gamma_1^-(u^2),\\
\tilde{\mathcal B}_2(u)&=&e^{\xi({\bf y}-\tilde{\partial}_{\bf x},u^2)}=\Gamma_2^+(u^2),\\
\tilde{\mathcal C}_2(u)&=&e^{\xi(\tilde{\partial}_{\bf y},u^{-2})}\  =\Gamma_2^-(u^2),
\eea
where the vertex operators $\Gamma_i^\pm(t),i=1,2$ are defined in (\ref{g1}) and (\ref{g2}).
\end{prp}
The A-model topological string partition function on $\C^3$ is given by the MacMahon function
\be
Z_{\C^3}^{top}=M(q)=\prod_{n=1}^\infty\frac{1}{(1-q^n)^n}=\sum_{n=0}^\infty P(n)q^n.
\ee
It is related to the topological vertex and $P(n)$ counts the number of plane partition whose total boxes number equals $n$. It is known that the generating function of plane partitions can be written as a fermionic correlator involving the standard vertex operators
\bea
\tilde{\Gamma}^+(z)=e^{\xi({\bf x},z)},\quad \tilde{\Gamma}^-(z)=e^{\xi(\tilde{\partial}_{\bf x},z^{-1})}
\eea
with a particular specialization of the values of $z=q^{\pm 1/2},q^{\pm 3/2}, q^{\pm 5/2},\cdots$.

In the following, we will give that $Z_{\C^3}^{top}$ can also be obtained from the vertex operators $\Gamma^{\pm}_i(t)$, $i=1,2$.
\begin{lem}\label{lemmazc}
The following relation holds
\be
\langle 0|\prod_{m=1}^\infty \Gamma_2^-(q^{m-1/2})\Gamma_1^-(q^{m-1/2})|0\rangle=\prod_{n\geq 1}(1-q^n)^n.
\ee
\end{lem}
\begin{proof}
Since
\[
\Gamma_1^-(w)=e^{\xi({\bf x}, w)}e^{-\xi(\tilde{\partial}_{\bf y}, w)},
\]
\[
\Gamma_2^-(z)=e^{\xi({\bf y}, z)}e^{-\xi(\tilde{\partial}_{\bf x}, z)},
\]
and
\[
e^{-\xi(\tilde{\partial}_{\bf x}, z)}e^{\xi({\bf x}, w)}=(1-zw)e^{\xi({\bf x}, w)}e^{-\xi(\tilde{\partial}_{\bf x}, z)}
\]
then we get
\[
\Gamma_2^-(z)\Gamma_1^-(w)=(1-zw):\Gamma_2^-(z)\Gamma_1^-(w):
\]
where the normal order is defined as usual. Using this formula step by step, we get the conclusion.
\end{proof}
\begin{prp}
The A-model topological string partition function on $\C^3$ (the MacMahon function) equals
\be\label{topstr}
Z_{\C^3}^{top}=\langle 0|\prod_{m=1}^\infty \Gamma_2^+(q^{-m+1/2})\Gamma_1^+(q^{-m+1/2})\prod_{m=1}^\infty \Gamma_2^-(q^{m-1/2})\Gamma_1^-(q^{m-1/2})|0\rangle.
\ee
\end{prp}
\begin{proof}
Since
\bea
\Gamma_i^+(z)\Gamma_i^-(w)&=&\frac{1}{1-w/z}\Gamma_i^-(w)\Gamma_i^+(z),\quad i=1,2,\nonumber\\
\Gamma_i^+(z)\Gamma_j^-(w)&=&\Gamma_i^j(w)\Gamma_i^+(z),\quad i,j=1,2,\ i\neq j.\nonumber
\eea
Then the right hand side of (\ref{topstr}) equals
\[
\frac{1}{(1-q^n)^n}\frac{1}{(1-q^n)^n}\langle 0|\prod_{m=1}^\infty \Gamma_2^-(q^{m-1/2})\Gamma_1^-(q^{m-1/2})|0\rangle.
\]
By the lemma \ref{lemmazc}, we get the conclusion.

\end{proof}

\section*{Acknowledgements}
The authors gratefully acknowledge the support of Professors Ke Wu, Zi-Feng Yang, Shi-Kun Wang.
Chuanzhong Li is supported by the National Natural Science Foundation
of China under Grant No. 11571192 and K. C. Wong Magna Fund in Ningbo University.

\end{document}